\documentclass[journal]{IEEEtran}
\usepackage{amsmath}
\usepackage{amsfonts,amsthm,amssymb}
\usepackage{xcolor}
\usepackage[font=scriptsize,caption=false,labelsep=space]{subfig}
\usepackage{graphicx,float,wrapfig}
\usepackage[none]{hyphenat}
\usepackage{url}
\newtheorem{prop}{Proposition}
\usepackage{balance}

\usepackage{cite}
\usepackage{comment}
\usepackage{booktabs}
\usepackage[named]{algo}
\usepackage[noend]{algpseudocode}
\usepackage{algorithm}
\usepackage{multirow}
\begin{document}

\title{Joint Transmit and Reflective Beamformer Design for Secure Estimation in IRS-Aided WSNs}
\author{Mohammad Faisal Ahmed, Kunwar~Pritiraj~Rajput,
       Naveen~K.~D.~Venkategowda, Kumar~Vijay~Mishra, and
        Aditya~K.~Jagannatham
        \thanks{M. F. A. is with Cisco Systems India Pvt. Ltd., Bengaluru
        560103, India, e-mail: mdfaisal165@gmail.com.} 
        \thanks{K. P. R. and A. K. J. are with the Department of Electrical Engineering, Indian Institute of Technology Kanpur, 208016 India, e-mail: \{pratiraj, adityaj\}@iitk.ac.in.}
        \thanks{N. K. D. V. is with the Department of Science and Technology, Link\"{o}ping University, Norrk\"{o}ping, 60174 Sweden, email: naveen.venkategowda@liu.se.}
        \thanks{K. V. M. is with the United States CCDC Army Research Laboratory, Adelphi, MD 20783 USA, e-mail: kvm@ieee.org.}
        }
\maketitle

\begin{abstract}
Wireless sensor networks (WSNs) are vulnerable to eavesdropping as the sensor nodes (SNs) communicate over an open radio channel. Intelligent reflecting surface (IRS) technology can be leveraged for  physical layer security in WSNs. In this paper, we propose a  joint transmit and reflective beamformer (JTRB) design for secure parameter estimation at the fusion center (FC) in the presence of an eavesdropper (ED) in a WSN. We develop a semidefinite relaxation (SDR)-based iterative algorithm, which alternately yields the transmit beamformer at each SN and the corresponding reflection phases at the IRS, to achieve the minimum mean-squared error (MSE) parameter estimate at the FC, subject to transmit power and ED signal-to-noise ratio constraints. Our simulation results demonstrate robust MSE and security performance of the proposed IRS-based JTRB technique.
\end{abstract}


\section{Introduction}
In wireless sensor networks (WSNs), the information transmitted by the sensor nodes (SNs) for inference tasks such as event detection, tracking, state estimation may be sensitive and must be kept secure from adversaries. The open radio channel communication allows eavesdroppers (EDs) to intercept data  transmitted by the SNs \cite{7258313,9540882,9324793}. Traditional cryptographic security is resource-intensive for use in low-power WSNs \cite{7258313,9071996,8758230}. Therefore, physical layer  security have been proposed as an low-complexity alternative for secure distributed detection \cite{Li_Secure_Detection_TSP19,7931677,Li_Secure_Detection_TSP20}
 and secure remote state estimation \cite{Ding_Estimation_Activeve_TAC21,Leong_Scheduling_Eavesdropper_TAC19,8646374} in WSNs.
 
 Physical layer secrecy exploits channel conditions and aims to improve the rate of delivering reliable information to a legitimate receiver, while also safeguarding the information from an ED. Seminal work  by Wyner \cite{Wyner75} introduced the wiretap channel and {\it secrecy capacity}, i.e., the maximum rate at which a legitimate receiver can correctly decode the source message, while no useful information about the source signal can be obtained by an ED. 
The efficacy of physical layer security depends on the channel characteristics. Hence, intelligent reflecting surfaces (IRS), which can reshape the wireless signal propagation environment through software-controlled reflecting surfaces, \cite{pfeiffer2013metamaterial,hodge2021deep} can be leveraged to significantly improve physical layer security \cite{Deligiannis_Secrecy_Radar_TAES18}.
An IRS is a two-dimensional surface consisting of a large number of passive meta-material elements to reflect the incoming signal through a pre-computed phase shift \cite{elbir2020survey}. IRS further improves security by reflecting the incident signal such that the interference is constructive at the intended receiver whereas destructive at the ED \cite{9071996,mishra2022opt}.

IRS-aided security often requires joint optimization of  transmit  beamformer and reflection surface phases. In \cite{9299783},  joint transmit and reflective beamforming (JTRB) was proposed to achieve a predefined secrecy rate while minimizing the total network power. The converse JTRB problem of maximizing the secrecy rate for a given transmit power budget was studied in \cite{8972400}, and subsequently for a multiple-input single-output (MISO) system in \cite{9014322}. A more challenging scenario, wherein the eavesdropping channel is stronger than the legitimate communication channel in addition to being spatially correlated, was considered in \cite{8723525}. A maximal ratio transmission-based MISO beamforming technique that minimizes the transmit power, subject to receiver signal-to-noise ratio (SNR) and ED jamming level constraints, was proposed in \cite{9146177}. To improve the quality of service (QoS) at users under jamming attacks in a multi-user system using IRS, robust beamforming when jammer's  channel state information (CSI) is unknown was proposed in \cite{9672153}. Sun \textit{et al.} \cite{9513582} consider IRS-assisted base station when both jammer and ED are present and develop an interesting JTRB design without the knowledge of jammer's transmit beamformer and ED's CSI. Some recent works \cite{Deligiannis_Secrecy_Radar_TAES18,eltayeb2017enhancing,su2020secure,fangsinr} also consider IRS-aided secrecy for dual-function sensing-communications \cite{mishra2019toward}.  

Contrary to the aforementioned works, where the objective was reliable reproduction of information at the receiver, we focus on maximizing the inference accuracy in WSN when noisy measurements from SNs are transmitted over unreliable channels. In this paper we consider the problem of remote estimation, where SNs transmit their observations to the fusion center (FC) that estimates a sensitive parameter. We employ IRS  not only for in-network signal processing of SN's measurements but also overcoming the wireless channel effects to minimize the estimation error while guaranteeing that ED cannot estimate the sensitive parameter with a specified accuracy.

In this context, there is a lack of research that considers IRS for facilitating secure parameter estimation at the FC from SNs observations transmitted over a coherent multiple access channel (MAC) \cite{7809037,8768026,8286913}. However, we propose a JTRB in IRS-aided WSN for secure parameter estimation by minimizing the mean-squared error (MSE) at the FC while requiring the MSE at the ED above a desired threshold. Assuming a minimum MSE estimator at the FC, the transmit beamformer at SNs and the reflection phases at the IRS is computed through a semi-definite relaxation (SDR)-based approach to solve the resulting non-convex optimization problem with power constraints. Numerical experiments demonstrate the efficacy of the proposed design and validate the analytical formulations.

Throughout the paper, we use $\text{Tr}[\cdot]$ and $\mathbb{E}\{\cdot\}$ to represent trace and statistical expectation operators, respectively; the notations $\odot$, $(.)^H$ and $|.|$ denote the Hadmard product, Hermitian, and magnitude, respectively; and $\mathbf{A}=D(\mathbf{a})$ is a diagonal matrix with the vector $\mathbf{a}$ as the principal diagonal. The notation $a \sim \mathcal{CN}(0,\sigma_a^2)$ denotes a circularly symmetric complex Gaussian random variable $a$ with zero mean and variance $\sigma_a^2$.

\section{System Model}
Consider a coherent MAC-based IRS-assisted WSN comprising of $K$ single-antenna SNs, an IRS with $N$ passive reflecting elements, and a single antenna FC  as shown in Fig.~\ref{Fig.SM}. The observations corresponding to each SN are transmitted to the FC, and can potentially be intercepted by an ED. Denote the unknown parameter to be estimated by $\theta \in \mathbb{C}$, with power $\mathbb{E}\left \{  | \theta |^2 \right \} = 1$. The $k$-th sensor observation is \par\noindent\small
\begin{align}
    x_{k} = \alpha _{k}\theta + n_{k},
\end{align}\normalsize
where $\alpha _{k} \in \mathbb{C}$ and $n_{k} \sim \mathcal{CN}(0,\sigma_k^2)$ represent the observation scaling factor and observation noise, respectively. 
Each SN precodes its observations using the precoding coefficient $\beta_{k}\in \mathbb{C}$, and subsequently, transmits them over a coherent MAC to the FC. Denote the wireless fading channel vector between the $k$-th SN and IRS by $\mathbf{h}_{k,I}\in \mathbb{C}^{N \times 1}$, with the concatenated channel matrix $\mathbf{H}_{I} = [\mathbf{h}_{1,I} , \hdots , \mathbf{h}_{K,I}]\in \mathbb{C}^{N \times K}$. The incident signal from all the SNs at the IRS is 
$\mathbf{y}_{I}=\sum_{k=1}^{K}\mathbf{h}_{k,I}\beta_{k}x_{k} = \mathbf{H}_{I}\tilde{\mathbf{x}}$, where
$\tilde{\mathbf{x}} = [\beta_{1}x_{1},\hdots,\beta_{K}x_{K}]^T$. 
Expanding $\mathbf{y}_{I}$ gives 
\begin{equation}
\mathbf{y}_{I} = \mathbf{H}_{I}((\boldsymbol{\alpha}\odot \boldsymbol{{\beta}})\theta + \boldsymbol{{\beta}}\odot\mathbf{n}),
\end{equation}
where the $K \times 1$ complex gain vector $\boldsymbol{\alpha}=[\alpha_1,\alpha_2,\hdots,\alpha_K]$, precoding vector $\boldsymbol{{\beta}}=[\beta_1,\beta_2,\hdots,\beta_K]$, and noise vector $\mathbf{n} = [n_{1},\hdots,n_{K}]^T$, with noise covariance $\mathbb{E}\left \{  \mathbf{n}\mathbf{n}^H \right \} = \sigma^2_{n}\textbf{I}$.

We model the passive IRS elements that reflect the incident signal, operated via a digital controller, using a diagonal phase-shift matrix $\boldsymbol{\Theta} = D(\boldsymbol{\phi})\in \mathbb{C}^{N \times N}$, where $\boldsymbol{\phi}=[e^{j\phi_{1}},\hdots,e^{j\phi_{N}} ] \in \mathbb{C}^{N \times 1}$. Define $\mathbf{h}_{I,F} \in \mathbb{C}^{1 \times N}$ as the channel vector between the IRS and FC, and $h_{k,f}\in\mathbb{C}$ as the scalar channel coefficient from the $k$-th SN to FC, with the concatenated channel vector $\mathbf{h}_{f} = [h_{1,f}, ... , h_{K,f}] \in \mathbb{C}^{1 \times N}$. The received signal at the FC is a superposition of the direct signals from all the SNs to the FC and the signal reflected by the IRS, i.e.,
\begin{equation}
{y}_{F} = (\mathbf{h}_{I,F}\boldsymbol{\Theta}\mathbf{H}_{I} + \mathbf{h}_{f})((\boldsymbol{\alpha}\odot \boldsymbol{{\beta}})\theta + \boldsymbol{{\beta}}\odot\mathbf{n} ) + n_{f},
\end{equation}
where $u_{f} \sim \mathcal{CN}(0,\sigma^2_{f})  \in \mathbb{C}$ is the  noise at the FC. 

Denote the channel vector between the IRS and ED by $\mathbf{h}_{I,E}\in \mathbb{C}^{1 \times N}$ and the scalar channel coefficient between the $k$-th SN and ED by $h_{k,e}\in \mathbb{C}$, with $\mathbf{h}_{e} = [h_{1,e}, ... , h_{K,e}]^T \in \mathbb{C}^{1 \times K}$. The received signal at the single-antenna ED is 
\begin{equation}
{y}_{E} = (\mathbf{h}_{I,E}\boldsymbol{\Theta}\mathbf{H}_{I} + \mathbf{h}_{e})((\boldsymbol{\alpha}\odot \boldsymbol{{\beta}})\theta + \boldsymbol{{\beta}}\odot\mathbf{n} ) + n_e,\; \in \mathbb{C},
\end{equation}
where $u_e \sim \mathcal{CN}(0,\sigma^2_e) \in \mathbb{C}$ is the noise at ED. Let ${\mathbf{{h}}}_F=(\mathbf{h}_{I,F}\boldsymbol{\Theta}\mathbf{H}_{I} + \mathbf{h}_{f})\in \mathbb{C}^{1 \times N}$ and $\mathbf{{h}}_{E}=(\mathbf{h}_{I,E}\boldsymbol{\Theta}\mathbf{H}_{I} + \mathbf{h}_{e})\in \mathbb{C}^{1 \times K}$. It follows that
${y}_{F} = \mathbf{{h}}_{F}((\boldsymbol{\alpha}\odot \boldsymbol{{\beta}})\theta + \boldsymbol{{\beta}}\odot\mathbf{n} ) + u_{f}$,
and
${y}_{E} = \mathbf{{h}}_{E}((\boldsymbol{\alpha}\odot \boldsymbol{{\beta}})\theta + \boldsymbol{{\beta}}\odot\mathbf{n}) + u_e$.

\begin{figure}
\includegraphics[width=0.75\linewidth]{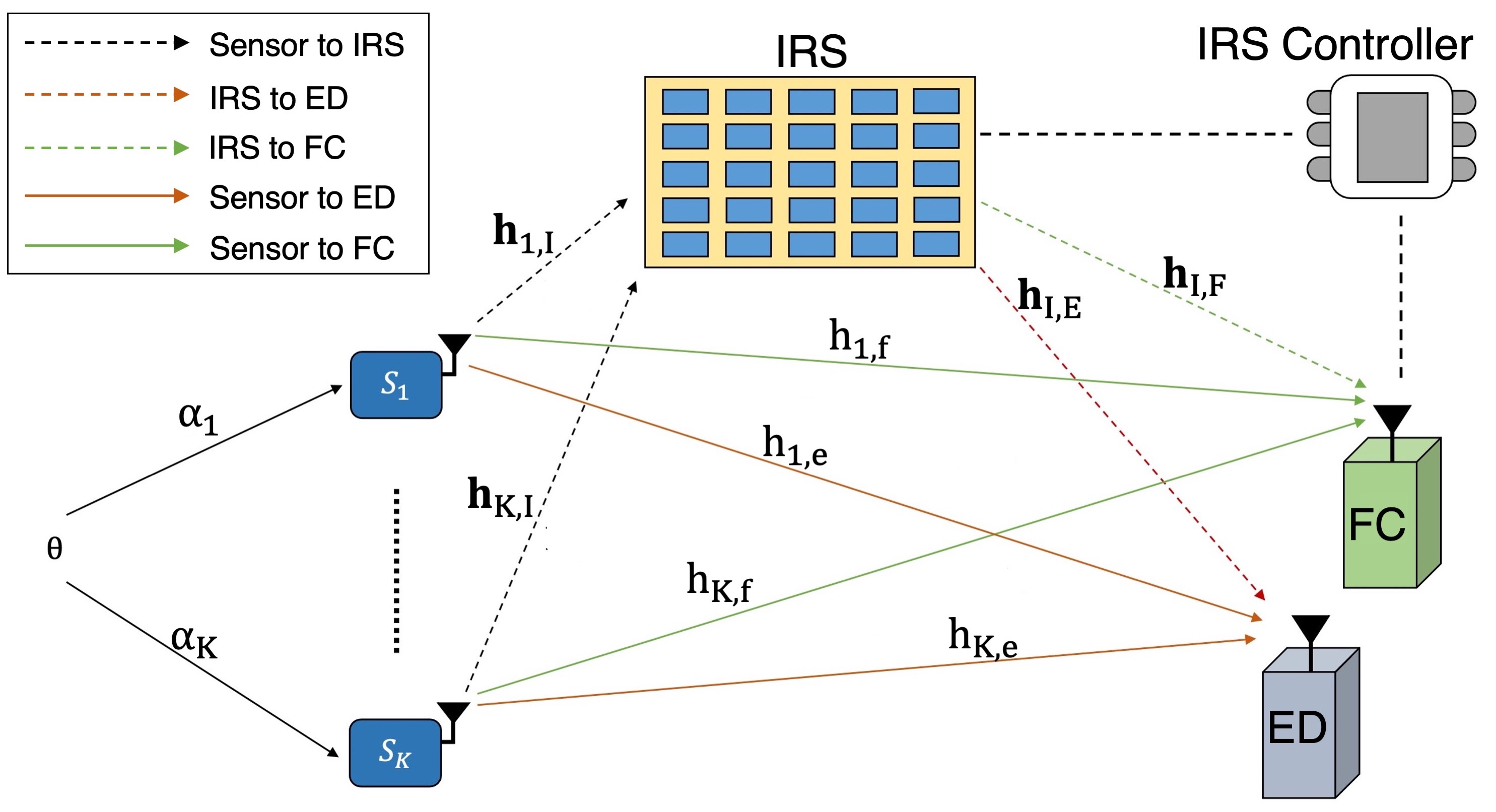}
\centering
\caption{IRS-based WSN with single-antenna sensor nodes transmitting measurements over a coherent MAC to the FC in the presence of an ED.}
\label{Fig.SM}
\end{figure}

On employing the linear minimum MSE (LMMSE) combiner at the FC and ED to generate the estimate $\hat{\theta}$ of the underlying parameter, the
resulting MSEs at the FC and ED, respectively, are\par\noindent\small
\begin{align}\label{EQ:MSE_FC}
\text{MSE}_{\text{FC}}=\Bigg(1+\underbrace{{\frac{\left | {\mathbf{h}}_{F}(\boldsymbol{\alpha} \odot\boldsymbol{{\beta}}) \right |^2}{\sigma_{o}^{2}\boldsymbol{{\beta}}^{H}D({\mathbf{h}}_F)D^{H}({\mathbf{h}}_F)\boldsymbol{{\beta}}+\sigma_{f}^{2}}}}_{\text{SNR}_{\text{FC}}}\Bigg)^{-1} 
\end{align}
\begin{align}\label{EQ:MSE_ED}
\text{MSE}_{\text{ED}}=\Bigg(1+\underbrace{\frac{\left | {\mathbf{h}}_{E}(\boldsymbol{\alpha} \odot\boldsymbol{{\beta}}) \right |^2}{\sigma_{o}^{2}\boldsymbol{\beta}^{H}D({\mathbf{h}}_E)D^{H}({\mathbf{h}}_E)\boldsymbol{{\beta}}+\sigma_{e}^{2}}}_{\text{SNR}_{\text{ED}}}\Bigg)^{-1}.
\end{align}\normalsize

It follows from \eqref{EQ:MSE_FC} that minimizing  $\text{MSE}_{\text{FC}}$ is equivalent to maximizing the $\text{SNR}_{\text{FC}}$. Also, to ensure that the ED is unable to reliably estimate the underlying parameter, one needs to design the JTRB such that the $\text{SNR}_{\text{ED}}$ value lies below a predefined threshold $\eta$. The total average network transmit power is constrained as
$\boldsymbol{{\beta}}^{H}D^{H}(\boldsymbol{\alpha})D(\boldsymbol{\alpha})\boldsymbol{{\beta}} + \sigma^{2}_{o}\boldsymbol{{\beta}}^{H}\boldsymbol{{\beta}} \leq P_T$, 
where $P_T$ is the transmit power budget. Subject to these constraints on the total power 
and $\text{SNR}_{\text{ED}}$, the optimization problem to minimize the $\text{MSE}_{\text{FC}}$ in \eqref{EQ:MSE_FC} is  \par\noindent\small
\begin{subequations}
\label{eq:opt}
\begin{align}
& \underset{\boldsymbol{{\beta}},\boldsymbol{\phi}}{\mathrm{maximize}}
& &   \frac{\left | {\mathbf{h}}_{F}(\boldsymbol{\alpha} \odot\boldsymbol{{\beta}}) \right |^2}{\sigma_{o}^{2}\boldsymbol{{\beta}}^{H}D({\mathbf{h}}_F)D^{H}({\mathbf{h}}_F)\boldsymbol{{\beta}} + \sigma^2_{f} } \label{11a}   \\
&\textrm{subject to}
& & \boldsymbol{{\beta}}^{H}D^{H}(\boldsymbol{\alpha})D(\boldsymbol{\alpha})\boldsymbol{{\beta}} + \sigma^{2}_{o}\boldsymbol{{\beta}}^{H}\boldsymbol{{\beta}} \leq P_{T} \label{11b}\\
& & & \frac{\left | {\mathbf{h}}_{E}(\boldsymbol{\alpha} \odot\boldsymbol{{\beta}}) \right |^2}{\sigma_{o}^{2}\boldsymbol{\beta}^{H}D({\mathbf{h}}_E)D^{H}({\mathbf{h}}_E)\boldsymbol{{\beta}}+\sigma_{e}^{2}} \leq \eta \label{11c}\\
& & & |\boldsymbol{\phi}_{i}| = 1 \;,\;i = 1,...,N. \label{11d}
\end{align}
\end{subequations} \normalsize
The coupling between $\boldsymbol{\beta}$ and $\boldsymbol{\phi}$ in the objective function in (7a)  renders it non-convex. Together, with the unit-modulus constraints, this problem is intractable. Hence, we develop an iterative scheme to solve this problem based on the principle of  alternate optimization and also prove that the proposed algorithm achieves convergence.

\section{JTRB Design Algorithm}

Define $\gamma$ as an auxiliary variable and recast \eqref{eq:opt} as
\par\noindent\small
\begin{align*}
\label{EQ:opt1}
& \underset{\boldsymbol{{\beta}},\boldsymbol{\phi},\gamma}{\mathrm{maximize}}
& &   \gamma \tag{8a}\\
&\mathrm{subject\; to} 
& &  \frac{\left | {\mathbf{h}}_{F}(\boldsymbol{\alpha} \odot\boldsymbol{{\beta}}) \right |^2}{\sigma_{o}^{2}\boldsymbol{{\beta}}^{H}D({\mathbf{h}}_F)D^{H}({\mathbf{h}}_F)\boldsymbol{{\beta}} + \sigma^2_{f} } \geq \gamma  \tag{8b} \label{12}\\
& & & \eqref{11b}, \eqref{11c}\; \text{and} \; \eqref{11d}.  
\end{align*}\normalsize
\setcounter{equation}{8}
The constraints in \eqref{11c} and \eqref{12} can be recast as
 $\text{T}(\boldsymbol{\phi},\boldsymbol{\beta})=\left | {\mathbf{h}}_{E}(\boldsymbol{\alpha} \odot\boldsymbol{{\beta}}) \right |^2 - \eta\left(\sigma_{o}^{2}{\mathbf{h}_E}D(\boldsymbol{{\beta}}^{H})D(\boldsymbol{{\beta}}){\mathbf{h}}^{H}_E - \sigma_{e}^{2}\right) \leq 0$, 
 and $\text{S}(\boldsymbol{\phi},\boldsymbol{\beta})= {\left | {\mathbf{h}}_{F}(\boldsymbol{\alpha} \odot\boldsymbol{{\beta}}) \right |^2} - \gamma(\sigma_{o}^{2}{\mathbf{h}_F}D(\boldsymbol{{\beta}}^{H})D(\boldsymbol{{\beta}}){\mathbf{h}}^{H}_F + \sigma_{f}^{2}) \geq 0$,  respectively. 
This changes the above optimization problem to\par\noindent\small 
\begin{align}
& \underset{\boldsymbol{\phi},\boldsymbol{\beta},\gamma}{\text{maximize}}
& & \gamma  \nonumber \\
& \text{subject to}
&  &\eqref{11b}, \eqref{11d},\text{T}(\boldsymbol{\phi},\boldsymbol{\beta}) \leq 0,\; 
\text{S}(\boldsymbol{\phi},\boldsymbol{\beta}) \geq 0. 
\end{align}\normalsize
To tackle this non-convex and, hence, intractable problem, we invoke SDR. Define $\mathbf{B}={\boldsymbol{{\beta}}}{\boldsymbol{{\beta}}}^H \in \mathbb{C}^{K \times K}$ with $\text{rank}(\mathbf{B}) \leq 1$
and $\boldsymbol{\Lambda}=D(\boldsymbol{\alpha}^{H})D(\boldsymbol{\alpha}) \in \mathbb{C}^{K \times K}$. The constraint in \eqref{11b} is
$\text{Tr}[\boldsymbol{\Lambda}\mathbf{B}]+\sigma^2_0\text{Tr}[\mathbf{B}] \leq P_T$.
Furthermore, the resulting problem is convex in each of the optimization variables, for fixed values of the other two.
The following proposition formulates the problem to determine the optimal value of $\boldsymbol{\phi}$ in the first iteration. 
\begin{prop}
\label{prop:SNR_Quad_Function}
For a given $\mathbf{B}$, $\text{S}(\boldsymbol{\phi},\mathbf{B}) = \text{S}(\boldsymbol{\phi}|\mathbf{B})$ is a quadratic function given by\par\noindent\small
\begin{flalign*}
   &\text{S}(\boldsymbol{\phi}|\mathbf{B})\nonumber\\
   &\hspace{-2.5mm}=\begin{bmatrix}
 \boldsymbol{\phi} \\1
 \end{bmatrix}^H \Bigg[\underbrace{\begin{bmatrix}
\mathbf{P}_1 & \mathbf{p}_2\\ 
\mathbf{p}^H_2 & 0
\end{bmatrix}}_{\mathbf{P}(\mathbf{B)}} -\gamma\sigma_{o}^{2} \underbrace{\begin{bmatrix}
\mathbf{R}_1 & \mathbf{r}_2\\ 
\mathbf{r}^H_2 & 0
\end{bmatrix}}_{\mathbf{R}(\mathbf{B)}}\Bigg]\begin{bmatrix}
 \boldsymbol{\phi} \\1
 \end{bmatrix} + p_3+r_3- \gamma\sigma_{f}^{2}.
\end{flalign*}\normalsize
\end{prop}
\begin{proof}
Substituting ${\mathbf{{h}}}_F=(\mathbf{h}_{I,F}\boldsymbol{\Theta}\mathbf{H}_{I} + \mathbf{h}_{f})$ in $\text{S}(\boldsymbol{\phi},\boldsymbol{\beta})= {\left | {\mathbf{h}}_{F}(\boldsymbol{\alpha} \odot\boldsymbol{{\beta}}) \right |^2} - \gamma(\sigma_{o}^{2}{\mathbf{h}_F}D(\boldsymbol{{\beta}}^{H})D(\boldsymbol{{\beta}}){\mathbf{h}}^{H}_F + \sigma_{f}^{2}) $ and using the properties of the diagonal operator and element-wise product, $ \text{S}(\boldsymbol{\phi}|\mathbf{B})$ is expressed in the stated quadratic form with ${\mathbf{P}_1} =D(\mathbf{{h}}_{I,F})\mathbf{H}_ID(\boldsymbol{\alpha})\mathbf{B}D(\boldsymbol{\alpha}^{H})\mathbf{H}^{H}_I D(\mathbf{{h}}^{H}_{I,F})$, $\mathbf{p}_2=D(\mathbf{{h}}_{I,F})  \mathbf{H}_ID(\boldsymbol{\alpha})\mathbf{B}D(\boldsymbol{\alpha}^{H})\mathbf{{h}}^{H}_{f}$, $p_3=\mathbf{{h}}_{f}D(\boldsymbol{\alpha})\mathbf{B}D(\boldsymbol{\alpha}^{H})\mathbf{{h}}^{H}_{f}$, 
$r_3=\mathbf{{h}}_{f}D(\mathbf{B})\mathbf{{h}}^{H}_{f}$, $\mathbf{r}_2=D(\mathbf{{h}}_{I,F})\mathbf{H}_ID(\mathbf{B})\mathbf{{h}}^{H}_{f}$, and $\mathbf{R}_1=D(\mathbf{{h}}_{I,F})\mathbf{H}_ID(\mathbf{B})\mathbf{H}^{H}_ID(\mathbf{{h}}^{H}_{I,F})$.
\end{proof}

Define $\boldsymbol{\tilde{\phi}}^H = [\boldsymbol{\phi}^H , 1]\in \mathbb{C}^{1 \times (N+1)}$ and $\mathbf{Q} = \boldsymbol{\tilde{\phi}}\boldsymbol{\tilde{\phi}}^H \in \mathbb{C}^{(N+1) \times (N+1)}$ where, $\text{rank}(\mathbf{Q}) \leq 1$. Using the identity $\mathrm{Tr}\left[\mathbf{A}\mathbf{B}\mathbf{C}\right]=\mathrm{Tr}\left[\mathbf{C}\mathbf{A}\mathbf{B}\right]$ yields\par\noindent\small
\begin{align*}
    &\text{S}(\mathbf{Q}|\mathbf{B}) =\mathrm{Tr}\left[\mathbf{Q}\mathbf{P}(\mathbf{B)}\right]+p_3-\gamma \left(\sigma_o^2\mathrm{Tr}\left[\mathbf{Q}\mathbf{R}(\mathbf{B)}\right]+r_3 +\sigma_f^2\right).
\end{align*}\normalsize
Similarly,\par\noindent\small
\begin{align*}
    &\text{T}(\boldsymbol{\phi}|\mathbf{B})=\boldsymbol{\tilde{\phi}}^H \widetilde{\mathbf{P}}(\mathbf{B)} \boldsymbol{\tilde{\phi}}+\tilde{p}_3-\eta \left(\sigma_o^2 \boldsymbol{\tilde{\phi}}^H  \widetilde{\mathbf{R}}(\mathbf{B)} \boldsymbol{\tilde{\phi}}+\tilde{r}_3 +\sigma_e^2\right)  \nonumber \\
    &=\mathrm{Tr}\left[\mathbf{Q} \widetilde{\mathbf{P}}(\mathbf{B)}\right]+\tilde{p}_3-\eta \left(\sigma_o^2\mathrm{Tr}\left[\mathbf{Q} \widetilde{\mathbf{R}}(\mathbf{B)}\right]+\tilde{r}_3 +\sigma_e^2\right) = \text{T}(\mathbf{Q}|\mathbf{B}),
\end{align*} \normalsize
where $\widetilde{\mathbf{P}}(\mathbf{B)}$, $\widetilde{\mathbf{R}}(\mathbf{B)}$, $\tilde{p}_3$, $\tilde{r}_3$ are, respectively, obtained from ${\mathbf{P}}({\mathbf{B}})$, ${\mathbf{R}}({\mathbf{B}})$, ${p}_3$, ${r}_3$ by replacing $\mathbf{h}_F$ ($\mathbf{h}_{I,F}$) by $\mathbf{h}_E$ ($\mathbf{h}_{I,E}$). Hence, relaxing the rank constraint, i.e., $\text{rank}(\mathbf{Q}) \leq 1$, the optimization problem for this scenario becomes \par\noindent\small
\begin{align}\label{Opt_B}
& \underset{\mathbf{Q}}{\text{maximize}}
& &  \gamma  \nonumber \\
& \text{subject to}
& &  \text{S}(\mathbf{Q}|\mathbf{B}) \geq 0, \text{T}(\mathbf{Q}|\mathbf{B}) \leq 0, \; \mathbf{Q} \succeq 0, \nonumber \\
& & & \mathbf{Q}_{i,i} = 1,\; \; i=1,\hdots,N+1.
\end{align}\normalsize
For a given value of $\gamma$ (obtained using the bisection method) and $\mathbf{B}$, \eqref{Opt_B} is a convex feasibility-check problem, which is readily solved via convex optimization tools such as CVX \cite{grant2008cvx}.\looseness=-1

Next, for a known $\mathbf{Q}$, using the relation ${\mathbf{{h}}}_F=(\mathbf{h}_{I,F}\boldsymbol{\Theta}\mathbf{H}_{I} + \mathbf{h}_{f})$, the term ${\left | {\mathbf{h}}_{F}(\boldsymbol{\alpha} \odot\boldsymbol{{\beta}}) \right |^2}$ reduces to \par\noindent\small
\begin{align}\label{EQ:19}
 &{\left | {\mathbf{h}}_{F}(\boldsymbol{\alpha} \odot\boldsymbol{{\beta}}) \right |^2} = \text{Tr}\Big[\underbrace{D(\boldsymbol{\alpha}^{H})\mathbf{{h}}_{f}}_{\mathbf{k}_2}\boldsymbol{{\phi}}^{H}D(\mathbf{{h}}_{I,F})  \mathbf{H}_ID(\boldsymbol{\alpha})\mathbf{B}\Big]  \nonumber \\
 &+\text{Tr}\Big[\underbrace{D(\boldsymbol{\alpha}^{H})\mathbf{H}^{H}_I D(\mathbf{{h}}^{H}_{I,F})}_{\mathbf{K}_1}\boldsymbol{{\phi}}\boldsymbol{{\phi}}^{H}\underbrace{D(\mathbf{{h}}_{I,F})\mathbf{H}_ID(\boldsymbol{\alpha})}_{\mathbf{K}^H_1} \mathbf{B}\Big]   \nonumber\\
 &+ \text{Tr}\Big[D(\boldsymbol{\alpha}^{H}) \mathbf{{h}}_{f}\mathbf{{h}}^{H}_{f}D(\boldsymbol{\alpha})\mathbf{B}\Big]+\text{Tr}\Big[D(\boldsymbol{\alpha}^{H})\mathbf{H}^{H}_ID(\mathbf{{h}}^{H}_{I,F})\boldsymbol{\phi}\mathbf{{h}}^{H}_{f}D(\boldsymbol{\alpha})\mathbf{B}\Big] \nonumber \\
&=\text{Tr}\bigg[\underbrace{\begin{bmatrix}
 \mathbf{K}_1 & \mathbf{k}_2
 \end{bmatrix}}_{\mathbf{K}} \underbrace{\begin{bmatrix}
\boldsymbol{{\phi}}\boldsymbol{{\phi}}^{H} & \boldsymbol{{\phi}} \\
\boldsymbol{{\phi}}^{H} & 1
\end{bmatrix}}_{\mathbf{Q}}
\begin{bmatrix}
 \mathbf{K}^H_1 \\ \mathbf{k}^H_2
 \end{bmatrix}\mathbf{B}\bigg] = \text{Tr}[\mathbf{K}\mathbf{Q}\mathbf{K}^{H}\mathbf{B}].
 \end{align}\normalsize
Furthermore, we recast the term ${\mathbf{h}_F}D(\boldsymbol{{\beta}}^{H})D(\boldsymbol{{\beta}}){\mathbf{h}}^{H}_F$ as\par\noindent\small
\begin{align}\label{EQ:20}
&{\mathbf{h}_F}D(\boldsymbol{{\beta}}^{H})D(\boldsymbol{{\beta}}){\mathbf{h}}^{H}_F \nonumber\\
&= \text{Tr}\left[\mathbf{B}^{T}D({\mathbf{h}}_F)D^{H}({\mathbf{h}}_F)\right] =  \text{Tr}\left[\mathbf{B}^{T}(({\mathbf{h}}_{F}^{H}{\mathbf{h}}_F)\odot\mathbf{I})\right], \nonumber\\
&= \text{Tr}\left[\mathbf{B}^{T}(((\mathbf{h}_{I,F}\boldsymbol{\Theta}\mathbf{H}_{I} + \mathbf{h}_{f})^H(\mathbf{h}_{I,F}\boldsymbol{\Theta}\mathbf{H}_{I} + \mathbf{h}_{f}))\odot\mathbf{I})\right], \nonumber \\
&=\text{Tr}\big[\mathbf{B}^{T}\big(\big(\underbrace{\mathbf{H}^H_{I}D(\mathbf{h}^H_{I,F})}_{\mathbf{L}_1}\boldsymbol{\phi}^{*}\boldsymbol{\phi}^{T}D(\mathbf{h}_{I,F})\mathbf{H}_{I} + \mathbf{H}^H_{I}D(\mathbf{h}^H_{I,F}) \nonumber \\
&\hspace{+13pt}\boldsymbol{\phi}^{*}\mathbf{h}_{f}+ \mathbf{h}^H_{f}\boldsymbol{\phi}^{T}D(\mathbf{h}_{I,F})\mathbf{H}_{I} + \mathbf{h}_{f}^{H}\mathbf{h}_{f}\big)\odot\mathbf{I}\big)\big], \nonumber \\
&=\text{Tr}\bigg[\mathbf{B}^{T}\bigg(\bigg(\underbrace{\begin{bmatrix}
 \mathbf{L}_1 & \mathbf{h}^{H}_f
 \end{bmatrix}}_{\mathbf{L}^{*}} \underbrace{\begin{bmatrix}
\boldsymbol{{\phi}}^{*}\boldsymbol{{\phi}}^{T} & \boldsymbol{{\phi}}^{*} \\
\boldsymbol{{\phi}}^{T} & 1
\end{bmatrix}}_{\mathbf{Q}^T}
\begin{bmatrix}
 \mathbf{L}^{H}_1 \\ \mathbf{h}_f
 \end{bmatrix}\bigg)\odot\mathbf{I}\bigg)\bigg] \nonumber \\
 &=\text{Tr}\left[\mathbf{B}^T((\mathbf{L}^*\mathbf{Q}^T\mathbf{L}^{T})\odot\mathbf{I})\right]=\text{Tr}\left[((\mathbf{L}\mathbf{Q}\mathbf{L}^{H})\odot\mathbf{I})\mathbf{B}\right].
 \end{align}\normalsize

Using \eqref{EQ:19} and \eqref{EQ:20}, we obtain 
 $\text{S}(\mathbf{B}|\mathbf{Q})= \text{Tr}[\mathbf{K}\mathbf{Q}\mathbf{K}^{H}\mathbf{B}] - \gamma(\sigma^2_{o}\text{Tr}[((\mathbf{L}\mathbf{Q}\mathbf{L}^{H})\odot\mathbf{I})\mathbf{B}] + \sigma^2_f)$. 
 Similar steps for $\text{T}(\mathbf{B}|\mathbf{Q})$ yield
$\text{T}(\mathbf{B}|\mathbf{Q})=\text{Tr}[\widetilde{\mathbf{K}}\mathbf{Q}\widetilde{\mathbf{K}}^{H}\mathbf{B}] - \eta(\sigma^2_{o}\text{Tr}[((\widetilde{\mathbf{L}}\mathbf{Q}\widetilde{\mathbf{L}}^{H})\odot\mathbf{I})\mathbf{B}] + \sigma^2_e)$,
where $\widetilde{\mathbf{K}}$ and $\widetilde{\mathbf{L}}$ are obtained from ${\mathbf{K}}$ and ${\mathbf{L}}$, respectively, by replacing $\mathbf{h}_F$ ($\mathbf{h}_{I,F}$) by $\mathbf{h}_E$ ($\mathbf{h}_{I,E}$). Hence, by relaxing the constraint $\text{rank}(\mathbf{B}) \leq 1$, the problem for this case changes to\par\noindent\small
\begin{align}\label{Opt_Q}
\underset{\mathbf{B}}{\text{maximize}}
\; &\gamma   \nonumber \\
 \text{subject to} & \;
 \text{Tr}[\boldsymbol{\Lambda}\mathbf{B}]+\sigma^2_0\text{Tr}[\mathbf{B}] \leq P_T,\nonumber\\
\;\;\;\;\;\;\;\;\;\;\;& \text{S}(\mathbf{B}|\mathbf{Q}) \geq 0, \; \text{T}(\mathbf{B}|\mathbf{Q}) \leq 0,\; \text{and} \; \mathbf{B} \succeq 0.
\end{align}\normalsize
Again, for a given $\gamma$ (obtained, as described previously, using the bisection method) and $\mathbf{Q}$, the problem in \eqref{Opt_Q} is also a convex feasibility-check problem that is solved similar to \eqref{Opt_B}. This iterative procedure is summarized in Algorithm~\ref{alg:label}, where the superscripts $(1)$ and $(2)$ denote variables corresponding to problems \eqref{Opt_B} and \eqref{Opt_Q}, respectively. The iterations are repeated till convergence or the maximum number of iterations. Note that if any of the obtained optimal matrices $\mathbf{B}^*$ and $\mathbf{Q}^*$ is rank-1, then the eigenvalue decomposition (EVD) of the respective matrix gives the optimal transmit or reflective beamformer. In case the rank of the solution is greater than 1, we apply the Gaussian-approximation technique \cite{5447068}, \cite{wu2019intelligent} to obtain a high quality rank-one solution invoking the EVD. The computational complexity of the proposed SDR-based JTRB design follows from \cite{gong2021beamforming} and is bounded by 
$\mathcal{O}(N_{iter}(log_{2}((\gamma_{max} - \gamma_{min})/\epsilon))(\sqrt{(2N + 4)}(N+1)^{6} + \sqrt{(K+3)}K^{6}))$.

\begin{algorithm}[H]
	\caption{SDR-based alternate optimization for JTRB} 
	\label{alg:label}
	\begin{algorithmic}[1]
		\Statex \textbf{Input:} $\boldsymbol{\Lambda}, \mathbf{P}(\mathbf{B)}, \mathbf{R}(\mathbf{B)}, \widetilde{\mathbf{P}}(\mathbf{B)}, \widetilde{\mathbf{R}}(\mathbf{B)}, \mathbf{L}, \mathbf{K}, \widetilde{\mathbf{L}}, \widetilde{\mathbf{K}}, p_3, r_3, \tilde{p}_3, \tilde{r}_3 ,$
		\Statex $\eta, P_T, \sigma^2_{e}, \sigma^2_{f}$ 
		\Statex \textbf{Output:} $\mathbf{Q}^\ast, \mathbf{B}^\ast$
        \State \textbf{Initialize:} $\gamma_{min}$, $\gamma_{max}$, $N_{iter}$, $\epsilon$, $\mathbf{{Q}} = \mathbf{{Q}}^{0}$, $n = 1$
        \For{$n = 1:N_{iter}$}
        \State $\gamma_{min}^{(1)} \leftarrow \gamma_{min}, \gamma_{max}^{(1)} \leftarrow \gamma_{max}$
        \While{$(\gamma_{max}^{(1)}-\gamma_{min}^1 \geq \epsilon)$}
        \State $\gamma_{opt}^{(1)} \leftarrow (\gamma_{max}^{(1)} + \gamma_{min}^{(1)})/2$,
        $\gamma \leftarrow \gamma_{opt}^{(1)}$ in \eqref{Opt_B}
        \State \textbf{if} \eqref{Opt_B} is feasible:
         $\gamma_{min}^{(1)} \leftarrow \gamma_{opt}^{(1)}$
         \textbf{else}:
         $\gamma_{max}^{(1)} \leftarrow \gamma_{opt}^{(1)}$
        \EndWhile
        \State \textbf{end while}
        \State 
        $\mathbf{Q}^{*} \leftarrow \mathbf{Q}(n)$
        \State 
        $\gamma_{min}^{(2)} \leftarrow \gamma_{min}, \gamma_{max}^{(2)} \leftarrow \gamma_{max}$
        \While{$(\gamma_{max}^{(2)}-\gamma_{min}^{(2)} \geq \epsilon)$}
        \State $\gamma_{opt}^{(2)} \leftarrow (\gamma_{max}^{(2)} + \gamma_{min}^{(2)})/2$,
        $\gamma \leftarrow \gamma_{opt}^{(2)}$ in \eqref{Opt_Q}
        \State \textbf{if} \eqref{Opt_Q} is feasible:
         $\gamma_{min}^{(2)} \leftarrow \gamma_{opt}^{(2)}$
         \textbf{else}:
         $\gamma_{max}^{(2)} \leftarrow \gamma_{opt}^{(2)}$
        \EndWhile
        \State \textbf{end while}
        \State 
        $\mathbf{B}^{*} \leftarrow \mathbf{B}(n)$
        \EndFor
        \State \textbf{end for}
        \State \textbf{return} $\mathbf{{B}}^\ast, \mathbf{{Q}}^\ast$ 
 \end{algorithmic}
\end{algorithm}
The following Proposition~\ref{prop:conv} proves that the proposed algorithm converges. 
\begin{prop}
\label{prop:conv}
The proposed SDR-based algorithm produces a monotonically non-decreasing sequence of $\gamma$ values and is guaranteed to converge to a local maximum.
\end{prop}
\begin{proof}
Assume, in the $i$-th iteration, the optimal objective value of (17) is $\gamma\left(\mathbf{B}(n),\mathbf{Q}(n)\right)$. The pertinent optimization problems in \eqref{Opt_B} and \eqref{Opt_Q} are convex in nature. It follows that
  $\gamma\left(\mathbf{B}(n),\mathbf{Q}(n)\right) \leq \gamma\left(\mathbf{B}(n+1),\mathbf{Q}(n)\right) \leq \gamma\left(\mathbf{B}(n+1),\mathbf{Q}(n+1)\right)$.
Since the objective value is upper bounded by $\gamma_{\text{max}}$, the proposed algorithm is guaranteed to converge.
\end{proof}

\section{Numerical Experiments}
\label{sec:numexp}
This section presents the results of a simulation-based study to validate the model and schemes described. In all experiments, we generate the small-scale fading coefficients of all the channels as i.i.d. samples of a $\mathcal{CN}(0,1)$ random variable. Hence, their envelopes are Rayleigh distributed. The path loss model is $\mu(\frac{d}{d_0})^{-\nu}$ where $\mu = -30$ dB is the path loss at the reference distance of $1$ m. The path loss exponent $\nu$ is set to $2$ for all SNs-to-IRS, IRS-to-FC, and IRS-to-ED links and $3$ for SNs to the FC and ED links. The locations of the SNs, IRS, FC and ED are given by their Cartesian coordinates. The SNs are distributed uniformly at random in the square region $([0,40] \times [0,40])$ m$^2$. The IRS, FC, and ED are located at coordinates $(60,20)$, $(65,25)$, and $(70,15)$ m, respectively. The noise variances $\sigma^2_{o}$, $\sigma^2_{e}$, and $\sigma^2_{f}$ are set to $-70$ dBm. We use $\epsilon = 0.01$ and $\gamma_{max} = \frac{\left | \mathbf{h}_{F} \right |^2P_{i}}{\sigma ^{2}_{f}}$.

Figure \ref{Fig2}a depicts the $\text{MSE}_{\text{FC}}$ performance as a function of the number of reflecting elements $N$ of the IRS with $\eta=1$ and $P_T=30$ dBm for $K=5$ and $8$ sensors. Note that, as $N$ increases, the $\text{MSE}_{\text{FC}}$ decreases because increasing $N$ leads to sharper reflective beamforming. Also, with the rise in the number of SNs within the WSN, the $\text{MSE}_{\text{FC}}$ performance improves due to the availability of more observations at the FC. With $N=10$ IRS elements and $K=5$ sensors, $\text{MSE}_{\text{FC}}$ reduces by $39\%$ over that of a non-IRS system, clearly illustrating the advantage of incorporating an IRS subsystem in the WSN. For the same number of sensors, when $N$ increases ten-fold (10 to 100), the $\text{MSE}_{\text{FC}}$ decreases by $78\%$, which is significant.

Figure \ref{Fig2}b demonstrates the $\text{MSE}_{\text{FC}}$ performance as a function of total power budget $P_T$ for different values of the number of reflecting elements $N$ of the IRS. As the transmit power in increased in the WSN, the MSE performance improves as expected. The MSE is also seen to decrease with $N$, similar to the trend observed previously. Figure \ref{Fig2}c shows the $\text{MSE}_{\text{FC}}$ performance against a varying threshold $\eta$ for $K=5$, $N=20$ and $P_T=30$ dBm. The  $\text{MSE}_{\text{FC}}$ improves when $\eta$ increases, since this enlarges the feasible region to determine the transmit and reflective beamformers. Once again, higher values of $P_T$ lead to lower $\text{MSE}_{\text{FC}}$. 

Finally, Fig.~\ref{Fig2}d depicts the convergence of the proposed SDR-based JTRB scheme with respect to the number of iterations with $\eta=1$ and $P_T=30$ dBm for $K=5$ and $N=20$. Note that the scheme converges within very few iterations, which makes it well-suited for practical implementation.
\begin{figure}[t]
\captionsetup[subfloat]{position=top,labelformat=empty}
\centering
\vspace{-16pt}
\hspace{-2pt}\subfloat[]{\includegraphics[width=0.9\columnwidth]{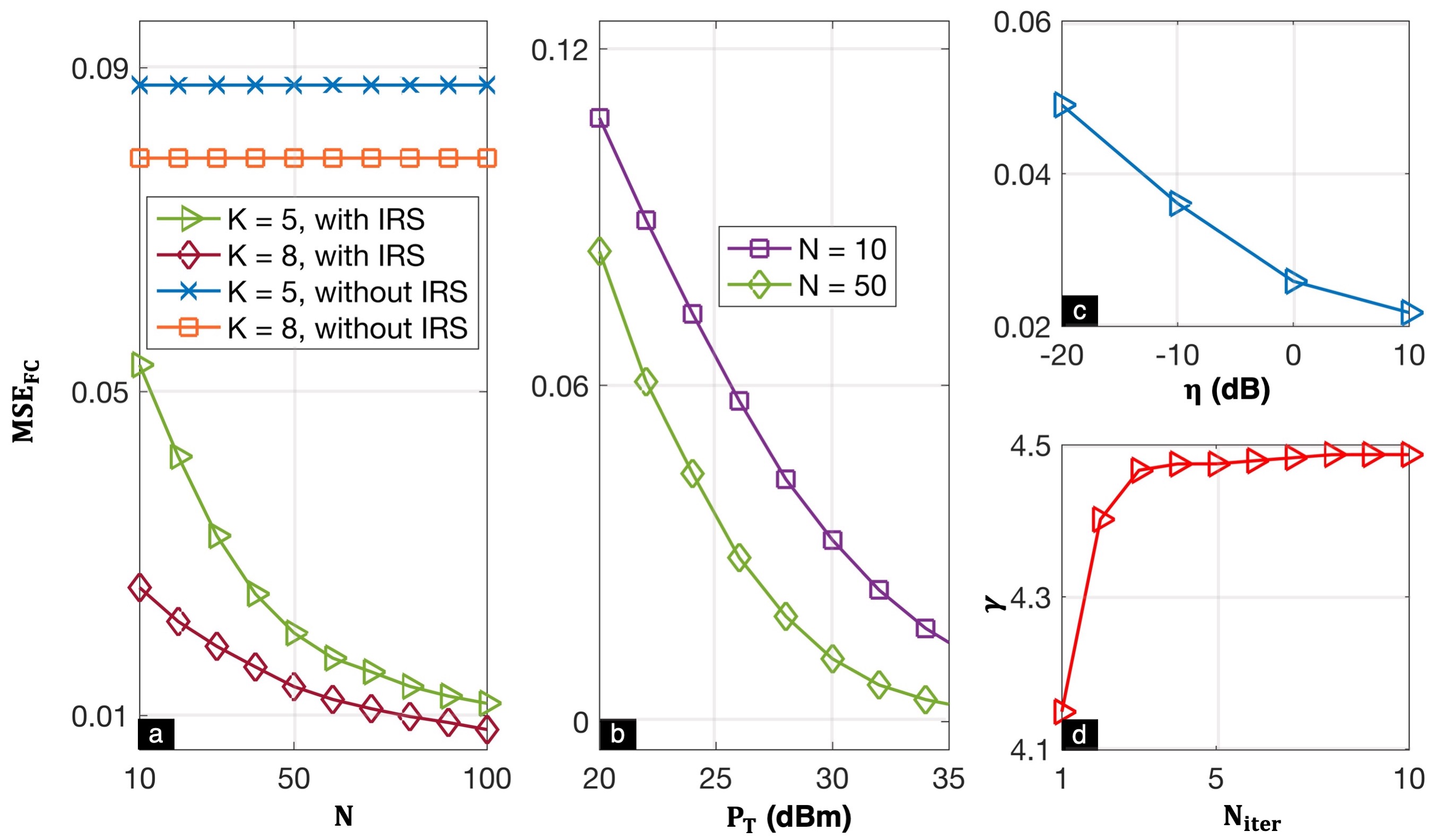}}
\hfil
\caption{$\left(a\right)$ $\text{MSE}_{\text{FC}}$ versus number of IRS elements $N$. $\left(b\right)$ $\text{MSE}_{\text{FC}}$ versus total power $P_T$. $\text{MSE}_{\text{FC}}$ versus ED's threshold $\eta$ for $\left(c\right)$  $P_T=35$ dBm and $\left(d\right)$ Objective value of (11a) versus the number of iterations.
}
\label{Fig2}
\end{figure}

\section{Summary}
This paper presented a novel JTRB design procedure for secure FC-based parameter estimation in an IRS-assisted WSN. The proposed iterative scheme minimized the $\text{MSE}_{\text{FC}}$ while obeying the total transmit power and ED SNR constraints. It was observed that the presence of an IRS results in a substantial improvement in the MSE - nearly $40$\% for 5 sensors and 10 IRS elements - at the FC, as compared to a non-IRS-based WSN, without compromising the security. 
\clearpage
\balance
\bibliographystyle{IEEEtran}
\bibliography{main}
\end{document}